\documentclass[showpacs,amsmath,amssymb,twocolumn,pra,superscriptaddress,notitlepage]{revtex4-1}
\usepackage[dvips]{graphicx} % for figures
\usepackage{amsfonts}
\usepackage{amssymb}
\usepackage{amscd}
\usepackage{amsmath}    % need for subequations
\usepackage{enumerate}
\usepackage{epsfig}
\usepackage{subfigure}
\usepackage{bm}
\usepackage{xcolor}
\usepackage{amsthm}
\usepackage{framed}
\usepackage{multirow}
\usepackage{mathrsfs,amssymb}
\usepackage{fancybox, graphicx}
\usepackage[skins]{tcolorbox}

\newtcolorbox[auto counter]{tbox}[2][]{%
    enhanced, float=hbt, drop fuzzy shadow southeast,
    colback=white!5!white, colframe=white!50!black,
    width= .97\columnwidth,sharp corners, boxrule=0.8pt,
    title={Table \thetcbcounter: #2}, #1
}

\newtheorem{theorem}{Theorem}
\newtheorem{lemma}{Lemma}

\newtheorem{definition}{Definition}

\newtheorem{proposition}{Proposition}

\newcommand{\bra}[1]{\mbox{$\left\langle #1 \right|$}}
\newcommand{\ket}[1]{\mbox{$\left| #1 \right\rangle$}}

\begin{document}
\title{Unification of quantum resources in distributed scenarios}
\begin{abstract}
Quantum resources, such as coherence, discord, and entanglement, play as a key role for demonstrating advantage in many computation and communication tasks. In order to find the nature behind these resources, tremendous efforts have been made to explore the connections between them. In this work, we extend the single party coherence resource framework to the distributed scenario and relate it to basis-dependent discord. We show the operational meaning of basis-dependent discord in quantum key distribution. By formulating a framework of basis-dependent discord, we connect these quantum resources, coherence, discord, and entanglement, quantitatively, which leads to a unification of measures of different quantum resources.

%We expect that the GPE-QKD protocol would become a standard in future QKD implementations.
\end{abstract}
\author{Hongyi Zhou}
\email{zhouhy14@mails.tsinghua.edu.cn}
\affiliation{Center for Quantum Information, Institute for Interdisciplinary Information Sciences, Tsinghua University, Beijing, 100084 China}

\author{Xiao Yuan}
\email{xiao.yuan.ph@gmail.com}
\affiliation{Department of Materials, University of Oxford, Parks Road, Oxford OX1 3PH, United Kingdom}

\author{Xiongfeng Ma}
\email{xma@tsinghua.edu.cn}
\affiliation{Center for Quantum Information, Institute for Interdisciplinary Information Sciences, Tsinghua University, Beijing, 100084 China}

\maketitle

%insert suggested PACS numbers in braces on next line

%\maketitle %\maketitle must follow title, authors, abstract and \pacs

\section{introduction}
Coherence, discord, and entanglement are fundamental resources in many tasks that cannot be achieved by classical physics.
Coherence characterizes the superpositions \cite{aberg2006quantifying,Baumgratz14}, serving as a resource of quantum randomness generation \cite{YuanPhysRevA2015,yuan2016interplay,PhysRevA.97.012302}, quantum metrology \cite{giovannetti2004quantum,PhysRevA.94.052324,escher2011general,giorda2017coherence}, quantum computation \cite{PhysRevA.93.012111,anand2016coherence,Ma16,matera2016coherent}, and quantum thermodynamics \cite{PhysRevLett.113.150402,PhysRevX.5.021001,PhysRevLett.115.210403,mitchison2015coherence,goold2016role,PhysRevLett.118.070601}.
As one of the most widely used quantum resources, entanglement \cite{PhysRevA.53.2046,PhysRevLett.78.5022,PhysRevLett.78.2275,PhysRevLett.80.2245,Horodecki09,PhysRevA.65.032314,plbnio2007introduction,horodecki2013quantumness} plays a key role in quantum teleportation \cite{PhysRevLett.70.1895}, quantum key distribution \cite{bb84,Ekert91}, and dense coding \cite{PhysRevLett.69.2881}, and also interprets the violation of Bell inequalities. Discord characterizes quantum correlations beyond entanglement \cite{PhysRevLett.88.017901,Maziero09,PhysRevLett.105.190502,PhysRevA.83.032324,gu2012observing,dakic2012quantum,Modi12}. It is the resource for remote state preparation \cite{dakic2012quantum}, and might explain the acceleration in discrete quantum computation with one qubit and other quantum computation circuits \cite{Datta09}.

Although these quantum resources play different roles in different tasks, the nature behind the resources might be the same. To find out such a non-classical nature, a natural idea is to build a unification framework of these quantum resources. Recently some researches have made progress for this goal \cite{Yao15,bu2017distribution,Hu17relative,yadin2016quantum}. Early researches in this field focus on the transformation between distillable entanglement and discord \cite{Piani11,Streltsov11}. Since the framework of coherence is proposed \cite{Baumgratz14}, there have been substantial attempts for unifying coherence and entanglement resource theory by designing protocols where these two resources can be converted into each other \cite{Streltsov15,qi2017measuring,chin2017generalized,PhysRevA.96.032316,PhysRevA.94.022329}. One example is that a single partite state with non-zero coherence is shown to be able to generate entanglement with bipartite incoherent operations \cite{Streltsov15}. Similar results are extended to discord and generalized to multipartite systems in \cite{Ma16}, where it is shown that the quantum discord created by multipartite incoherent operations is bounded by the quantum coherence consumed in its subsystems. Another connection between coherence and entanglement lies in quantum state merging \cite{horodecki2005partial}. A standard quantum state merging can lead to a gain of entanglement, while the incoherent quantum state merging \cite{Streltsov16} where one of the parties is restricted with local incoherent operations only, shows that entanglement and coherence cannot be gained at the same time.

All the works above are trying to connect part of these resources. Recently a unification of all the three resources based on an interferometric scenario is proposed in \cite{yuan2017unified}. Considering a phase encoding process of an input state, the interferometry power, i.e., how much phase information can be obtained is determined by the quantum resource contained in the input state. In such an interferometric framework, different quantum resources corresponds to the interferometry power in different scenarios.
Although coherence, discord, and entanglement are qualitatively unified in the interferometric framework, a quantitative unification is still an open problem.
%Whether the measures of the different resources

%We can further explore how to quantitatively unify the measures of these quantum resources.

%one resource will be quantitatively transformed into another. The mutual transformations of the measures of these resources will lead to a quantitative unification.

In this work, we construct such a quantitative unification of the three resources. We first review the general definitions of resource frameworks and summarize the corresponding definitions for coherence, discord, and entanglement. Then, we extend the single party coherence resource framework to the bipartite distributed scenario in several different ways. It turns out that one of the definitions is identical to basis-dependent (BD) discord \cite{yadin2016quantum,Ma16}.
We construct the resource framework of BD-discord, where we propose its operational meaning in quantum key distribution (QKD) and give some examples of BD-discord measures. With the help of BD-discord, measures of coherence, discord, and entanglement can be naturally defined and unified.
%Finally, we consider the operational meaning of the result in in quantum key distribution (QKD). When the communication partner know the explicit expression of their joint state, the BD-discord is exactly the raw key in QKD.
We believe our unified framework of quantum resources can make a substantial progress in understanding the quantum nature.

%We show that BD-discord is exactly the raw key if the communication partner know their joint state. We further show that the raw key becomes entanglement under an individual attack using our unification results, which is consistent with the conventional security analysis of QKD.

\section{Preliminaries}
In this section, we first review the definitions of a general resource framework. Then, we briefly summarize the coherence framework and refer the reader to Appendix~\ref{app:discordandentanglement} for a detailed review of discord and entanglement frameworks.

\subsection{Resource framework}
%[XXX] change it to a general resource framework
A general resource framework \cite{Baumgratz14,gour2015resource,PhysRevX.5.041008,du2015coherence,PhysRevLett.115.070503,PhysRevLett.116.120404,matera2016coherent,PhysRevLett.117.030401,Chitambar16,RevModPhys.89.041003,PhysRevLett.119.230401} consists of the definition of free state, free operation, and resource quantifiers.

\emph{Free state} is a set of states $\mathcal{F}$ that contain no resource while a state $\rho \notin \mathcal{F}$ contains resource.

\emph{Free operations} are physical realizable operations characterized by completely positive and trace preserving (CPTP) maps. They should at least transform free states only into free states, i.e, $\Lambda_{CPTP}(\rho)\in \mathcal{F}, \forall\rho\in \mathcal{F}$ which can be rewritten as $\sum_n K_n \rho K_n^\dag \in \mathcal{F}, \forall \rho\in \mathcal{F}$ in Kraus presentation. Here $\{K_n\}$ is the set of Kraus operators satisfying $\sum_n K_n^\dag K_n=I$. Different other free operations can be defined based on different extra requirements.

\emph{Quantifiers} are real-valued functions $f$ mapping states to non-negative real numbers. The free states should be mapped to zero, i.e., $f(\rho)=0, \forall \rho\in \mathcal{F}$. And for an arbitrary state, the function value should not increase under free operations, i.e., $f(\rho)\geq f(\Lambda_{CPTP}(\rho))$. Other principles are required for different resources and different tasks.

\subsection{Framework of coherence}
The general resource framework reduces to a specific one when we consider coherence, discord, and entanglement as the resource.
We briefly review the coherence framework introduced in ~\cite{Baumgratz14,RevModPhys.89.041003}, focusing on quantum states in a $d$-dimensional Hilbert space.

\emph{Incoherent and maximally coherent states.} Given a classical computational basis $J= \{\ket{j}\}$, $(j=1,2,\dots,d)$, an
incoherent state refers to a state without superposition on the basis, which can be described by
\begin{equation}\label{eq:incoherentstate}
 \sigma = \sum_{j=1}^d p_{j}\ket{j}\bra{j},
\end{equation}
where $p_{j_A}\in[0,1], \forall j$ and $\sum_j p_{j_A} = 1$.
At the meantime, maximally coherent states can be expressed as:
\begin{equation}\label{eq:maximallycoherentstate}
\ket{\Psi_d}=\frac{1}{d}\sum_{j=1}^d e^{i\phi_j}\ket{j},
\end{equation}
where $\phi_j\in[0,2\pi)$.

\emph{Incoherent operations.} Incoherent operations map an incoherent state only to an incoherent state. That is,
$\sum_n\hat{K}_n\rho\hat{K}_n^\dag \subset \mathcal{C}, \forall \rho \in \mathcal{C}$, where $\mathcal{C}$ is the set of incoherent states, $\{\hat{K}_n\}$ is a series of
Kraus operators satisfying $\sum_n \hat{K}_n^\dag\hat{K}_n=I$.

%Specifically, as an analogy of the definition of local operations and classical
%communication (LOCC), the incoherent completely positive and trace preserving (ICPTP) quantum operations can be defined as:
%\begin{equation}\label{eq:icptpmap}
%\Phi_{\mathrm{ICPTP}}(\rho)=\sum_n p_n\rho_n,
%\end{equation}
%where $p_n = \mathrm{Tr}[\hat{K}_n \rho \hat{K}_n^\dag]$ and $\rho_n =  \hat{K}_n \rho \hat{K}_n^\dag/p_n$

\emph{Coherence measures.} A coherence measure $C(\rho)$ is defined by a function that maps a quantum states $\rho$ to a real non-negative number, which satisfies the following
conditions in Table.~\ref{Fig:coherenceProperties}:

\begin{tbox}[label=Fig:coherenceProperties]{Properties of a coherence quantifier.}
\begin{enumerate}[(C1)]
\item
$C(\sigma)=0$ when $\sigma$ is an incoherent state. A stronger condition is (C1') $C(\sigma)=0$ if and only if $\sigma$ is an incoherent state;
\item
\emph{Monotonicity}: Coherence should not increase under incoherent operations, that is, (C2a) $C(\rho)\geq C[\Phi_{\mathrm{ICPTP}}({\rho})]$ , (C2b)
$C(\rho) \geq \sum_n p_n C(\rho_n)$, where $\rho_n=K_n \rho K_n^\dag/\mathrm{tr}(K_n \rho K_n^\dag)$;
\item
\emph{Convexity}: Coherence cannot increase under mixing, that is, $\sum_e p_e C(\rho_e) \geq C\left(\sum_e p_e \rho_e\right)$.
\end{enumerate}
\end{tbox}
We leave the framework of the other two quantum resources, discord and entanglement in Appendix~\ref{app:discordandentanglement}. %states and free operations together with those of coherence in Table.~\ref{Tab:resourcesum}
%\begin{tbox}[hbt]\label{Tab:resourcesum}
%\begin{tabular}{c|ccc}
%\hline
	%Resources & coherence & discord & entanglement \\
	%\hline
	%Free states & incoherent states & classical-quantum states & separable states \\
	%Free operations & incoherent operations & local operations & LOCC \\
	%\hline	
%\end{tabular}
%\caption{Brief summary of coherence, discord and entanglement.}
%\end{tbox}
%\begin{table}[htb]
%\begin{framed}
%\centering
%\begin{enumerate}[(C1)]
%\item
%$C(\sigma)=0$ when $\sigma$ is an incoherent state. A stronger condition is (C1') $C(\sigma)=0$ if and only if $\sigma$ is an incoherent state;
%\item
%\emph{Monotonicity}: Coherence should not increase under incoherent operations, that is, (C2a) $C(\rho)\geq C[\Phi_{\mathrm{ICPTP}}({\rho})]$ , (C2b)
%$C(\rho) \geq \sum_n p_n C(\rho_n)$;
%\item
%\emph{Convexity}: Coherence cannot increase under mixing, that is, $\sum_e p_e C(\rho_e) \geq C\left(\sum_e p_e \rho_e\right)$.
%\end{enumerate}
%\end{framed}
%\caption{Properties that a coherence measure should satisfy.} \label{Fig:coherenceProperties}
%\end{table}

\section{Extending coherence to the distributed scenario}
Quantum coherence is defined in the single party scenario while discord and entanglement are defined for at least two parties. Therefore, to unify the three measures, we should generalize coherence to multiple parties. In this section, we consider three approaches to generalize coherence to the bipartite distributed scenario, where we begin with three possible generalized definitions of the incoherent state.

\subsection{Incoherent-incoherent bipartite coherence}
A natural extension is the bipartite coherence proposed in \cite{Streltsov15}, which considers the joint basis $J_AJ_B=\{\ket{j_A}\ket{j_B}\}$ $(j_A=1,2,\dots, d_A, j_B=1,2,\dots,d_B)$ with $d_A$ and $d_B$ being dimensions of the local Hilbert spaces of system $A$ and $B$, respectively.
The bipartite incoherent state in can be rewritten as
\begin{equation}\label{eq:iistate}
\sigma_{AB}^{II}=\sum_{j_A,j_B}p_{j_Aj_B}\ket{j_A}\bra{j_A}\otimes\ket{j_B}\bra{j_B}.
\end{equation}
It is not hard to see that the bipartite incoherent state defined above is a specific type of classical-classical state $\sigma_{AB}^{CC}=\sum_{m,n}p_{mn}\ket{m}\bra{m}\otimes\ket{n}\bra{n}$ with certain local bases. Here we call Eq.~\eqref{eq:iistate} as \emph{incoherent-incoherent (II) state}. A bipartite state contains bipartite coherence if it is not an incoherent-incoherent state.

\subsection{Incoherent-classical bipartite coherence}
When focusing the coherence in a local basis of system $A$ (say $J_A$) and ignore the local basis of system $B$, we define the \emph{incoherent-classical  (IC)} state as
\begin{equation}
\sigma_{AB}^{IC}=\sum_{j_A,n}p_{j_An}\ket{j_A}\bra{j_A}\otimes\ket{n}\bra{n}	
\end{equation}
which is still a classical-classical state. Although $\ket{j_A}$ is still from the $J_A$ basis, $\ket{n}$ can be from an arbitrary basis of the system $B$.
Note that any incoherent-classical state can be obtained by applying a local unitary operation on system $B$ to a incoherent-incoherent state, i.e., $\sigma_{AB}^{IC}=U_B \sigma_{AB}^{II} U_B^\dag$. Therefore, the set of incoherent-classical states is larger than the set of incoherent-incoherent states.
A bipartite state contains incoherent-classical bipartite coherence if it is not an incoherent-classical state.

%Here classical state in party $B$ means a mixture of a set of orthogonal bases, which can be realized with a certain unitary operation on the incoherent bases $J_B$, i.e.,

\subsection{Incoherent-quantum bipartite coherence}
In the above generalization, we still consider the incoherent state as a classical-classical state. If we only focus on the  coherence in a local basis of system $A$ (say $J_A$), and totally ignore the other party ($B$), we can generalize coherence to be
incoherent-quantum (IQ) coherence \cite{bu2017distribution, Chitambar16assisted},
\begin{equation}\label{eq:iqstate}
	\sigma_{AB}^{IQ}=\sum_{j_A=1}^{d_A}p_{j_A}\ket{j_A}\bra{j_A}\otimes \rho_B^{j_A}.
\end{equation}
Equivalently, it can be written as
\begin{equation}\label{eq:iqstate1}
\sigma_{AB}^{IQ}=\sum_{j_A}p_{j_A}\ket{j_A}\bra{j_A}\otimes \left(\sum_{l_{j_A}}l_{j_A}\ket{l_{j_A}}\bra{l_{j_A}}\right),	
\end{equation}
with a spectral decomposition in party $B$. It is not hard to see that incoherent-quantum state can be obtained by mixing incoherent-classical states. Or we can regard the set of incoherent-quantum state as the convex hull of the set of the incoherent-classical states. A bipartite state contains incoherent-quantum bipartite coherence if it is not an incoherent-quantum state.
%Consider an incoherent-incoherent state with the same spectrum in $J_B$, $\sigma_{AB}^{II}=\sum_{j_A}p_{j_A}\ket{j_A}\bra{j_A}\otimes (\sum_{l_{j_A}}\lambda_{l_{j_A}}\ket{j_B}\bra{j_B})$ and a unitary operation depending on $j_A$, $U_{j_A}\ket{j_B}=\ket{l_{j_A}}$, we can see that an incoherent-quantum state can be realized with a set of probabilistic unitary operations on an incoherent-incoherent state, the probability depends on the decomposition probability of the incoherent-quantum state $p_{j_A}$.

We generalize the bipartite coherence in distributed scenarios with the track of $\sigma_{AB}^{II}\rightarrow\sigma_{AB}^{IC}\rightarrow\sigma_{AB}^{IQ}$. The incoherent-incoherent state is a subset of incoherent-classical state which is further a subset of incoherent-quantum state. We illustrate the relationship of these states in Fig.~\ref{fig:classification2}. The incoherent-quantum bipartite coherence is actually identical to the basis-dependent (BD) discord \cite{yadin2016quantum,Ma16}, which is the key resource for our unification framework.
\begin{figure}[hbt]
\centering
\includegraphics[width=8 cm]{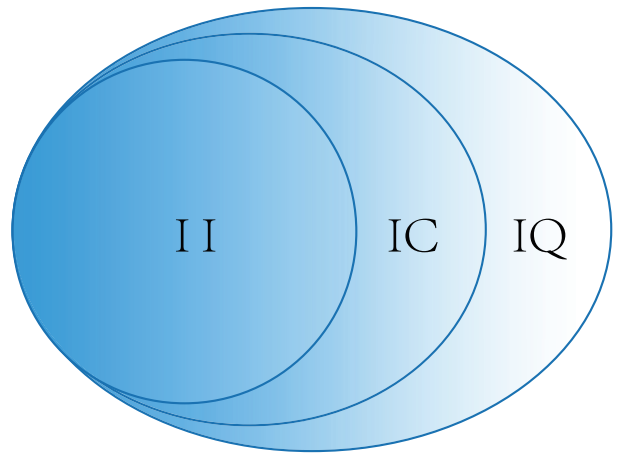}
\caption{Illustration of free states for different types of bipartite coherence. $II$: incoherent-incoherent states; $IC$: incoherent-classical states; $IQ$: incoherent-quantum states.} %A logarithmic function $\log_2{(2\mu+1)}$ will be between the two curves when $\mu$ is large enough.
\label{fig:classification2}
\end{figure}
%\paragraph{BD-discord}

\section{basis-dependent discord}
\subsection{Framework of basis-dependent discord}
The concept of BD-discord has been proposed in \cite{yadin2016quantum,Ma16} when studying discord. Here we formulate its resource framework, beginning with defininitions of free states for BD-discord given a local computational basis $J_A=
\{\ket{j_A}\} (j=1,2,\dots,d_A)$ on system A.

\begin{definition}
A zero basis-dependent discord state in $J_A=
\{\ket{j_A}\}$ is an incoherent-quantum state in Eq.~\eqref{eq:iqstate}
\end{definition}
%Meanwhile, a locally maximally quantum state can be expressed as:
%\begin{equation}\label{eq:locallymaximallycoherentstate}
%\ket{\Psi}=\ket{\Psi_d}\otimes \ket{\phi_B}= \left(\frac{1}{d}\sum_{j_A=1}^d e^{i\phi_{j_A}}\ket{j_A}\right)\otimes \ket{\phi_B},
%\end{equation}
%where $\phi_{j_A}\in[0, 2\pi]$ and $\ket{\phi_B}$ is an arbitrary state of system $B$. Note that, in general, the state of system $B$ can also be mixed as long as the state of system $A$ is maximally coherent. Besides the form of Eq.~\eqref{eq:locallymaximallycoherentstate}, we will show as follows that there exist other maximally quantum states that is not separable.

Second we define free operations for BD-discord, which map incoherent-quantum states to incoherent-quantum states.
\begin{definition}
The free operations for BD-discord are separable-quantum-incoherent (SQI) operations \cite{Streltsov17}
\begin{equation}
\Lambda_{SQI}(\sigma^{IQ}_{AB})=\sum_n{\hat{A}_n\otimes \hat{B}_n}\sigma_{AB}{\hat{A}^\dag_n\otimes \hat{B}^\dag_n} \subset \delta_{IQ},
\end{equation}
where $\delta_{IQ}$ is the set of incoherent-quantum states,
${\hat{A}_n\otimes \hat{B}_n}$ is a series of Kraus operators satisfying the completeness condition $\sum_n \hat{A}_n^\dag  \hat{A}_n\otimes \hat{B}_n^\dag \hat{B}_n=I$, and $\{\hat{A}_n\}$ is a set of incoherent operations on $A$.
\end{definition}

Finally we define the measures of BD-discord,  $BD_{J_A}\left(\rho_{AB}\right)$, which map a bipartite quantum states $\rho_{AB}$ to a real
non-negative number, satisfying the conditions in Table \ref{Fig:quantumProperties}.

\begin{tbox}[label=Fig:quantumProperties]{Properties of a basis-dependent discord quantifier.}
\begin{enumerate}[(BD1)]
\item
Basis-dependent discord vanishes for incoherent-quantum state $\sigma^{J_A}_{AB} = \sum_{j_A=1}^{d_A} p_{j_A}\ket{j_A}\bra{j_A} \otimes \rho^{j_A}_{B}$
\item
\emph{Monotonicity}: Basis-dependent discord should not increase under SQI operations, i.e., $  BD_{J_A}(\Lambda_{SQI}(\rho_{AB}))\le BD_{J_A}(\rho_{AB})$
\item
Basis-dependent discord is invariant under a local incoherent unitary operation on $A$ and a unitary operation on $B$
\end{enumerate}
\end{tbox}

\subsection{Examples of basis-dependent discord measures}
Here we give two categories of BD-discord measures that fulfill the conditions in Table.~\ref{Fig:quantumProperties}. One is the distance-based measure. The BD-discord equals to the distance from IQ states, which is expressed as
\begin{equation}\label{eq:distancemeasure}
	BD_{J_A}(\rho_{AB})=\min_{\sigma_{AB}^{IQ}\in \delta_{IQ}}d(\rho_{AB}|| \sigma_{AB}^{IQ}).
\end{equation}
Specifically, the distance can be various of measures given in Table~\ref{tab:bddiscordmeasures}, where the superscript in $\rho_{AB}^{Adiag}$ means a local dephasing operation on $A$. Actually these measures are widely used in entanglement, discord and coherence.

The other is the convex roof of local randomness,
\begin{equation}\label{eq:convexmeasure}
	BD_{J_A}(\rho_{AB})=\min_{p_e, \ket{\Psi_{AB}}_e}\sum_e p_e R(\ket{\Psi_{AB}}_e)
\end{equation}
where the minimization is over all possible pure state decompositions of $\rho_{AB}$, and $R(\ket{\Psi_{AB}}_e)$ is the local randomness given by von Neumman entropy of party $A$
\begin{equation}\label{Eq:localrandomness}
	R(\ket{\Psi_{AB}}_e)=S(\sum_{j_A}\bra{j_A}\mathrm{tr}_B(\rho_{AB})\ket{j_A}\ket{j_A}\bra{j_A})
\end{equation}
We prove that Eq.~\eqref{eq:distancemeasure} and Eq.~\eqref{eq:convexmeasure} satisfy all conditions of a BD-discord measure in Appendix~\ref{app:proofmeasure}.
%The relationship between free states for basis-dependent discord, discord and entanglement has been investigated in \cite{yuan2017unified}. The classical-quantum states are unions of incoherent-quantum states $\sigma_{AB}^{CQ}=\bigcup\sigma_{AB}^{IQ}$ over all local bases $J_A$, while the separable states are convex combinations of incoherent-quantum states $\sigma_{AB}^{sep.}=\sum_{\lambda} \lambda \sigma_{AB}^{IQ}$. Such relationships are also shown in Fig.~\ref{fig:classification1}. Inspired by this, we give the quantitative unifications of measures of basis-dependent discord, discord and entanglement, and leave the detailed proofs in Appendix.

\begin{tbox}[label=tab:bddiscordmeasures]{Some possible measures of basis-dependent discord.}
\begin{enumerate}[(1)]
\item
relative entropy, $S(\rho_{AB}^{Adiag})-S(\rho_{AB})$
\item
$l_1$ norm, $\underset{\sigma_{AB}^{IQ}\in \delta_{IQ}}\min||\rho_{AB}-\sigma_{AB}^{IQ} ||_{l_1}$
\item
geometric measure, $1-\underset{{\sigma_{AB}^{IQ}\in \delta_{IQ}}}\max F(\rho_{AB},\sigma_{AB}^{IQ})$
\item
fidelity measure, $1-\underset{{\sigma_{AB}^{IQ}\in \delta_{IQ}}}\max\sqrt{F(\rho_{AB},\sigma_{AB}^{IQ})}$
\end{enumerate}
\end{tbox}

\subsection{Operational meaning of the basis-dependent discord}
In this section, we consider the operational meaning of BD-discord, which is the local randomness of the raw key in QKD.
In the QKD security analysis, the communication partners, Alice and Bob, share a bipartite state $\rho_{AB}$, while the adversary Eve, is assumed to hold a purification $\ket{\Psi_{ABE}}$ of Alice's and Bob's system $AB$, which enables her to obtain the most information. The Devetak-Winter formula \cite{devetak2005distillation} gives an asymptotic key rate with one-way direct reconciliation. When $\rho_{AB}$ is known to Alice and Bob, the formula is expressed as
\begin{equation}\label{eq:devetakwinter1}
K=S(Z_A|E)-S(Z_A|Z_B)
\end{equation}
where $S(\cdot)$ is the von Neumann entropy function and $Z_{A(B)}$ is a local key generation measurement expressed as $\{\ket{j_{A(B)}}\bra{j_{A(B)}}\}$, $j_{A(B)}=1,2,\cdots d_{A(B)}$. We will show that the first term in Eq.~\eqref{eq:devetakwinter1} is actually a basis-dependent discord measure.
\begin{proposition}
The local randomness in QKD, i.e., the conditional entropy $S(Z_A|E)$ in the Devetak-Winter formula, is a BD-discord measure.
\end{proposition}
\begin{proof}
The conditional entropy $S(Z_A|E)$ can be expressed as
\begin{equation}\label{eq:conditionalentropy}
S(Z_A|E)=S(\rho_{AE}^{Adiag})-S(\rho_E).
\end{equation}
Suppose the tripartite state after Alice's local measurement is $\rho_{ABE}^{Adiag}=\sum_{j_A} p_{j_A} \ket{j_A}\bra{j_A} \otimes \rho_{BE}^{j_A}$, where $\rho_{BE}^{j_A}$ is a pure state since $\ket{\Psi_{ABE}}$ is a pure state, then
\begin{equation}
\begin{aligned}
\rho_{AB}^{Adiag}& =\mathrm{tr}_E({\rho_{ABE}^{Adiag}}) \\
& =\sum_{j_A} p_{j_A} \ket{j_A}\bra{j_A} \otimes  \mathrm{tr}_E(\rho_{BE}^{j_A}) \\
& =\sum_{j_A} p_{j_A} \ket{j_A}\bra{j_A} \otimes  \rho_B^{j_A}
\end{aligned}
\end{equation}
and $\rho_{AE}^{Adiag}$ has a similar expression of $\rho_{AE}^{Adiag}=\sum_{j_A} p_{j_A} \ket{j_A}\bra{j_A} \otimes  \rho_E^{j_A}$. Consider the von Neumann entropy of a classical-quantum state,
\begin{equation}
\begin{aligned}
S(\rho_{AE}^{Adiag})& =H(\{p_{j_A}\})-\sum_{j_A}p_{j_A}S(\rho_{E}^{j_A}) \\
& =H(\{p_{j_A}\})-\sum_{j_A}p_{j_A}S(\rho_{B}^{j_A}) \\
& =S(\rho_{AB}^{Adiag}),
\end{aligned}
\end{equation}
where $H(\cdot)$ is the Shannon entropy function and the second equality uses the fact that $S({\rho_B^{j_A}})$=$S({\rho_E^{j_A}})$ when $\rho_{BE}^{j_A}$ is a pure state, then Eq.~\eqref{eq:conditionalentropy} becomes
\begin{equation}\label{eq:relativeentropybddiscord}
\begin{aligned}
S(Z_A|E)&=S(\rho_{AB}^{Adiag})-S(\rho_{E})\\
&=S(\rho_{AB}^{Adiag})-S(\rho_{AB})\\
&=BD_{J_A}(\rho_{AB})
\end{aligned}
\end{equation}
where the last equation is the relative entropy measure of basis-dependent discord given in Table.~\ref{tab:bddiscordmeasures}.
\end{proof}

\section{Unifying measures of quantum resources}
With the help of the framework of BD-discord, now we are ready to unify the measures of different quantum resources.
\subsection{BD-discord to coherence}
%\subsection{Alternative definitions of basis dependent discord}
In previous section, BD-discord is extended from bipartite coherence. And now we redefine the original single partite coherence \cite{Baumgratz14} with BD-discord.

%First we extend the measures of single party coherence to that of the BD discord.  The BD discord is the resource of the assisted coherence scenario \cite{Chitambar16} where a bipartite state $\rho_{AB}$ is hold by two parties Alice and Bob, and Bob helps Alice to generate coherence in party $A$ as much as possible. Bob can perform arbitrary quantum operations while Alice is restricted to perform incoherent operations \cite{Chitambar16}. Under such a scenario, we can extend any coherence measure to a BD discord measure as described by the k  following theorem.

\begin{theorem}\label{Theorem:coherence}
%A coherence monotone

The BD-discord measure of a tensor product state $\rho_A\otimes \rho_B$ is a coherence monotone of $\rho_A$, i.e.,
\begin{equation}\label{eq:bddiscordtocoherence}
C_{J_A}(\rho_A)=BD_{J_A}(\rho_A \otimes \rho_B)
\end{equation}
For simplicity, we can calculate the coherence of $\rho_A$ by $BD_{J_A}(\rho_A \otimes I_B)$, where $I_B$ is an identity matrix of $B$. If $BD_{J_A}(\rho_A \otimes \rho_B)$ is further convex over $\rho_A \otimes \rho_B$, $C_{J_A}(\rho_A)$ becomes a coherence measure.
\end{theorem}

\begin{proof}
First, for an incoherent state $\sigma_A=\sum_{j_A}p_{j_A}\ket{j_A}\bra{j_A}$, $\sigma_A\otimes \rho_B =\sum_{j_A}p_{j_A}\ket{j_A}\bra{j_A} \otimes \rho_B$ is an IQ state. Then the rhs of Eq.~\eqref{eq:bddiscordtocoherence} equals to zero, which means $C_{J_A}(\sigma_A)=0$ for an incoherent state $\sigma_A$.

%IQ states $\sigma_{AB}^{IQ}$ given in Eq.~\eqref{eq:iqstate}, $f(\sigma_{AB}^{IQ})=0$. The partial trace operation belongs to $\Lambda_{SQI}$. And $f(\mathrm{tr}_B(\sigma_{AB}^{IQ}))\leq f(\sigma_{AB}^{IQ})=0$, which means $f(\sum_{j_A}p_{j_A}\ket{j_A}\bra{j_A})=0$.

Second, according to the contractivity of a BD-discord measure under $\Lambda_{SQI}$, $BD_{J_A}(\Lambda_{SQI}(\rho_A \otimes \rho_B))=BD_{J_A}(\Lambda_{IO}(\rho_A) \otimes \Phi(\rho_B)) \leq BD_{J_A}(\rho_A \otimes \rho_B)$, where $\Lambda_{IO}$ is an incoherent operation and $\Phi$ is an arbitrary operation. Then $C_{J_A}(\Lambda_{IO}(\rho_A))\leq C_{J_A}(\rho_A)$, which means $C_{J_A}(\rho_A)$ is contractive under $\Lambda_{IO}$.

%from the conditions of a BD-discord measure we have $f(\Lambda_{SQI}(\rho_{AB}))\leq f(\rho_{AB})$. If we take partial trace (notated as $\Lambda_{SQI}^0$) on both sides of the inequality, then the inequality still holds since $f(\Lambda_{SQI}^0 \circ \Lambda_{SQI}(\rho_{AB}))\leq f(\Lambda_{SQI}^0 (\rho_{AB}))$, i.e., $f(\Lambda_{IO}(\rho_A))\leq f(\rho_A)$ where $\Lambda_{IO}(\rho_A)=\sum_n \hat{A}_n \rho_A \hat{A}_n^\dag$ is an incoherent operation on $A$.

Finally, If $BD_{J_A}(\rho_A \otimes \rho_B)$ is convex over $\rho_A \otimes \rho_B$, i.e., $BD_{J_A}(\rho_A \otimes \rho_B) \leq \sum_n p_n BD_{J_A}(\rho_A^n \otimes \rho_B^n)$, where $\rho_A \otimes \rho_B=\sum_n p_n \rho_A^n \otimes \rho_B^n$. Then $C_{J_A}(\rho_A) \leq \sum_n p_n C_{J_A}(\rho_A^n)$, which shows the convexity of $C_{J_A}(\rho_A)$. A coherence monotone with convexity is a coherence measure.
\end{proof}

%We can see that basis-dependent discord is also related to single partite coherence.
\subsection{BD-discord to discord}
Furthermore, we can define a discord measure from any BD-discord measure. The free state for discord we consider here is the classical-quantum state, i.e.,
\begin{equation}
  \sigma^{CQ}_{AB} = \sum_n p_n\ket{n}\bra{n}\otimes\rho_B^n,
\end{equation}
where $\{\ket{n}\}$ is orthogonal for different $n$, $p_n\in[0,1], \forall n$ and $\sum_n p_n = 1$. As the set of classical-quantum state contains all the incoherent-quantum state in different local bases, one can regard discord as a basis-independent version of BD-discord. Based on such an intuition, we can define a discord measure by Theorem~\ref{Theorem:discord}.
\begin{theorem}\label{Theorem:discord}
A discord measure is a minimization of BD-discord measure over local bases, i.e.,
\begin{equation}\label{eq:discord}
%\begin{aligned}
	D(\rho_{AB})=\min_{U_A}BD_{J_A}(U_A\otimes I \rho_{AB}U^\dag_A \otimes I)
%\end{aligned}
\end{equation}
\end{theorem}
We leave the proof in Appendix \ref{app:proofdiscord}.

\subsection{BD-discord to entanglement}
To define entanglement measures from BD-discord measures, we consider the strong adversary scenario in \cite{yuan2017unified}. For a given input state $\rho_{AB}$, some phase information is encoded in the local basis $J_A$, i.e., by a local operation of $U_{A}=\sum_{j_A=1}^{d_A}e^{i\phi_{j_A}}\ket{j_A}\bra{j_A}$. After the phase encoding, a joint measurement is performed on both $A$ and $B$ to extract the phase information. It turns out that the interferometry power, i.e., how much phase information can be extracted corresponds to the BD-discord of the input state $\rho_{AB}$. A strong adversary holds a purification of $\rho_{AB}$ with $\rho_{AB}=\mathrm{tr}_E(\ket{\Psi}\bra{\Psi}_{ABE})$. In order to let the extracted phase information as little as possible, the adversary will choose an optimal measurement on her local quantum system $E$ and rotate the phase-encoding basis according to the measurement results. In this case the interferometry power corresponds to entanglement. Since the the local measurement on $E$ will effectively make the remaining system be with a certain decomposition $\rho_{AB}=\sum_e p_e \ket{\psi_{AB}}\bra{\psi_{AB}}_e$ and the basis rotation operation depends on $e$, the interferometry power will be minimized over all kinds of decompositions and the local unitary operations on $A$. Therefore we have the following theorem.
\begin{theorem}\label{Theorem:entanglement}
	An entanglement measure is a convex roof of a discord measure, i.e.,
\begin{equation}\label{eq:entanglementdefinition}
\begin{aligned}
		& E(\rho_{AB})= \\
		&\min_{p_e,\ket{\psi_{AB}}_e}\sum_e p_e \min_{U_A^e} BD_{J_A}(U^e_A\otimes I\ket{\psi_{AB}}\bra{\psi_{AB}}_eU^{\dag e}_A \otimes I)
%&=\min_{p_e,\ket{\psi_{AB}}_e}\sum_e p_e \min_{U^e_A} \min_{\sigma^{IQ}_{AB}\in \delta_{IQ}}d(U^e_A\otimes I \ket{\psi_{AB}}_e\bra{\psi_{AB}}_e U^{\dag e}_A \otimes I|| \sigma^{IQ}_{AB})
\end{aligned}
\end{equation}
	where the minimization is over all possible decompositions of $\rho_{AB}=\sum_e p_e \ket{\psi_{AB}}\bra{\psi_{AB}}_e$ and $\ket{\psi_{AB}}_e$ is a pure state.
\end{theorem}
We leave the proof in Appendix \ref{app:proofentanglement}.

\subsection{Example with distance-based measures}

%\subsection{Unification of quantum resources with a distance-based basis-dependent discord measure}

%\appendix

In this section we show an example of the measure unification of different quantum resources, the distance-based measures. Given distance-based BD-discord in Eq.~\eqref{eq:distancemeasure}, the distance-based coherence, discord and entanglement measures are given by Theorem~\ref{Theorem:coherence}, Theorem~\ref{Theorem:discord} and Theorem~\ref{Theorem:entanglement}
\begin{widetext}
\begin{equation}\label{eq:BDdiscordtodiscordandentanglement}
	\begin{aligned}
	C_{J_A}(\rho_{A}) & =\min_{\sigma^{IQ}_{AB}\in \delta_{IQ}}d( \rho_{A} \otimes \rho_B|| \sigma^{IQ}_{AB}) \\
		D(\rho_{AB})&=\min_{U_A}\min_{\sigma^{IQ}_{AB}\in \delta_{IQ}}d(U_A\otimes I \rho_{AB}U^\dag_A \otimes I|| \sigma^{IQ}_{AB}) \\
		E(\rho_{AB})&=\min_{p_e,\ket{\psi_{AB}}_e}\sum_e p_e \min_{U^e_A} \min_{\sigma^{IQ}_{AB}\in \delta_{IQ}}d(U^e_A\otimes I \ket{\psi_{AB}}_e\bra{\psi_{AB}}_e U^{\dag e}_A \otimes I|| \sigma^{IQ}_{AB}).
	\end{aligned}
\end{equation}
\end{widetext}

We note that the unification results will also be applied for other measures. And the operational meanings of each resource will be consistent in our unification framework. For example, the relative entropy measure of BD-discord will be transformed into distillable coherence, discord and entanglement by \eqref{eq:BDdiscordtodiscordandentanglement} which are also quantified by relative entropy.
%\begin{theorem}\label{Theorem:discord}
%A discord measure is a minimization of basis-dependent discord measure over local bases, i.e.,
%\begin{equation}\label{eq:discord}
%\begin{aligned}
	%D(\rho_{AB})&=\min_{U_A}BD_{J_A}(U_A\otimes I \rho_{AB}U^\dag_A \otimes I) \\
%%\end{aligned}
%\end{equation}
%\end{theorem}

%\begin{lemma}\label{lemma:sep}
%A separable state is a convex combination of incoherent-quantum states in different local basis, i.e.,
%\begin{equation}
	%\sum_{j=1}^d p_{j}\rho^{j}_A \otimes \rho^{j}_B=\sum_{k}\lambda_{k}\left(\sum_{{j}=1}^d p_{j}\ket{n_{{j}k}}\bra{n_{{j}k}} \otimes \rho^{j}_B\right)
%\end{equation}
%where $\rho^{j}_A = \sum_{k}\lambda_{jk}\ket{n_{jk}}\bra{n_{jk}}$ is a spectral decomposition.
%\end{lemma}

%\begin{theorem}\label{Theorem:entanglement}
	%An entanglement measure is a convex roof of a discord measure, i.e.,
%\begin{equation}\label{eq:entanglementdefinition}
%\begin{aligned}
		%E(\rho_{AB})&=\min_{p_e,\ket{\psi_{AB}}_e}\sum_e p_e D(\ket{\psi_{AB}}_e) \\
%&=\min_{p_e,\ket{\psi_{AB}}_e}\sum_e p_e \min_{U^e_A} \min_{\sigma^{IQ}_{AB}\in \delta_{IQ}}d(U^e_A\otimes I \ket{\psi_{AB}}_e\bra{\psi_{AB}}_e U^{\dag e}_A \otimes I|| \sigma^{IQ}_{AB})
%\end{aligned}
%\end{equation}
	%where the minimization is over all possible decompositions of $\rho_{AB}=\sum_e p_e \ket{\psi_{AB}}_e\bra{\psi_{AB}}_e$ and $\ket{\psi_{AB}}_e$ is a pure state.
	
%\end{theorem}

\section{Discussion and conclusion}

In this work, we propose a unification framework on coherence, basis-dependent discord, discord and entanglement. We begin with constructing a resource framework of basis-dependent discord. As a bridge, basis-dependent discord connects coherence for their basis-dependence nature. On the other hand, it relates discord and entanglement since they all characterize bipartite quantum correlations. A unification framework of these quantum resources is established with the help of BD-discord. Moreover, we give the operational meanings of basis-dependent discord in QKD, which correspond to the local randomness of keys.

For future work, it is interesting to generalize these results to continuous variable cases, especially for Gaussian states. Discord and entanglement for Gaussian states have been well defined based on covariance matrix presentations \cite{Horodecki2007,Modi12}, however, the quantum coherence or a coherence-like basis-dependent quantity is still missing. This work can provide an inspiration to complete the unifications of quantum resources for continuous variables. And this will also help us understand the quantum resource behind the secure keys in continuous variable QKD.

{\textbf{Acknowledgement}}
We acknowledge Y. Zhou and X. Zhang for the insightful discussions.
This work was supported by the National Natural Science Foundation of China Grants No.~11674193 , the National Key R\&D Program of China (2017YFA0303900, 2017YFA0304004), the National Research Foundation (NRF), NRF-Fellowship (Reference No: NRF-NRFF2016-02), BP plc and the EPSRC National Quantum Technology Hub in Networked Quantum Information Technology (EP/M013243/1).

H.Z. and X.Y. contributed equally to this work.
%\newpage

\appendix
\section{Framework of discord and entanglement}\label{app:discordandentanglement}
\subsection{Discord}

In this part, we briefly review the framework for quantum discord \cite{Henderson01,Ollivier01,Modi12} in a bipartite system $AB$.

\emph{Definition of classical state.}
A state is classical for discord when it is a classical-quantum state, i.e.,
\begin{equation}\label{Eq:CQstate}
  \sigma^{CQ}_{AB} = \sum_n p_n\ket{n}\bra{n}\otimes\rho_B^n,
\end{equation}
where $\{\ket{n}\}$ is orthogonal for different $n$, $p_n\in[0,1], \forall n$ and $\sum_n p_n = 1$.

\emph{Definition of classical operation.}
The classical operation for discord is defined by local operations on $B$, i.e., $I_A \otimes \Phi_B$.

\emph{Discord measure.} A discord measure $D(\rho_{AB})$ is defined by a function that maps a quantum states $\rho$ to a real non-negative number, which satisfies the following
conditions in Table~\ref{Fig:discordProperties}:

\begin{tbox}[label=Fig:discordProperties]{Discord properties.}
\begin{enumerate}[(D1)]
\item
$D\left(\sigma_{AB}\right)$ vanishes for classical-quantum states, $\sigma_{AB} = \sum_n p_n\ket{n}\bra{n}\otimes\rho_B^n$.
\item
\emph{Monotonicity}: $D\left(\rho_{AB}\right)$ cannot increase under local operations, $D(I_A \otimes \Phi_B(\rho_{AB}))\le D(\rho_{AB}) $.
\item
$D\left(\rho_{AB}\right)$ is invariant under all local unitary operations, $D(\rho_{AB})=D(U_A \otimes U_B \rho_{AB} U_A^\dag \otimes U_B^\dag)$.
\end{enumerate}
\end{tbox}

%\begin{table}[htb]
%\begin{framed}
%\centering
%\begin{enumerate}[(D1)]
%\item
%$D\left(\sigma_{AB}\right)$ vanishes for classical-quantum state, $\sigma_{AB} = \sum_n p_n\ket{n}\bra{n}\otimes\rho_B^n$.
%\item
%\emph{Monotonicity}: $D\left(\rho_{AB}\right)$ cannot increase under separable operations, $D(I_A \otimes \Phi_B(\rho_{AB}))\le D(\rho_{AB}) $.
%\item
%$D\left(\rho_{AB}\right)$ is invariant under all local unitary operations, $D(\rho_{AB})=D(U_A \otimes U_B \rho_{AB} U_A^\dag \otimes U_B^\dag)$.
%\end{enumerate}
%\end{framed}
%\caption{Properties that a discord measure should satisfy.} \label{Fig:discordProperties}
%\end{table}

\subsection{Entanglement}
In this part, we summarize the framework for entanglement \cite{Bennett96, Vedral98, Horodecki09} in a bipartite system $AB$.

\emph{Definition of classical state.}
A state is classical for entanglement when it is separable, i.e.,
\begin{equation}\label{}
  \sigma^{sep.}_{AB} = \sum_n p_n\rho_A^n\otimes\rho_B^n,
\end{equation}
where $p_n\in[0,1], \forall n$ and $\sum_n p_n = 1$.

\emph{Definition of classical operation.}
The classical operation for entanglement is defined by local operation and classical communication (LOCC). In the following, we denote LOCC operations by $\Lambda_{LOCC}$.

\emph{Entanglement measure.} An entanglement measure $D(\rho_{AB})$ is defined by a function that maps a quantum states $\rho$ to a real non-negative number, which satisfies the following
conditions in Table~\ref{Fig:entanglementProperties}:

\begin{tbox}[label=Fig:entanglementProperties]{Entanglement properties.}
\begin{enumerate}[(E1)]
\item
$E(\rho_{AB})$ vanishes when $\rho_{AB}$ is separable.
\item
\emph{Monotonicity}: $E(\rho_{AB})$ cannot increase under LOCC operation, that is, (E2a) $E[\Lambda_{LOCC}(\rho_{AB})]\leq E(\rho_{AB})$. This condition is often replaced by another stronger one. (E2b) $E(\rho_{AB})$ should not increase on average under LOCC operations which map $\rho_{AB}$ to $\rho_{AB}^k$ with probability $p_k$, then $\sum_k p_k E(\rho_{AB}^k)\leq E(\rho_{AB})$.
\item
\emph{Convexity}: $E(\rho_{AB})$ decreases under mixing, $E(\sum_k p_k \rho_{AB}^k) \leq \sum_k p_k E(\rho_{AB}^k)$.
\item
$E(\rho_{AB})$ is invariant under all local unitary operations, that is,  $E(\rho_{AB})=E(U_A \otimes U_B \rho_{AB} U_A^\dag \otimes U_B^\dag)$.
\end{enumerate}
\end{tbox}\label{}

\section{Proofs for BD-discord measures}\label{app:proofmeasure}
In order to formulate the conditions of BD-discord measures, we investigate the properties of incoherent unitary operations.
\begin{lemma}\label{lemma:1}
The Kraus operator of a unitary operation is unique.
\end{lemma}
\begin{proof}
Consider a unitary operation $U$, one possible Kraus operator representation can be written as $UU^\dag$ which is rank 1. All of its other Kraus operator representations are $E_i=\sum_ju_{ij}U_j$, where $u_{ij}$ is a unitary matrix. Since $U$ is rank 1, the matrix $u_{ij}$ reduces to $1$ and $E_i=U$.
\end{proof}
\begin{lemma}\label{lemma:2}
If a unitary operation is an incoherent operation, its inverse operation is also an incoherent operation
\end{lemma}
\begin{proof}
Consider a unitary operation $U$, it has unique Kraus operator representation $UU^\dag$ according to Lemma~\ref{lemma:1}. If it is an incoherent operation,  we have
\begin{equation}\label{eq:lemma:2}
C(U\rho U^\dag)\le C(\rho)
\end{equation}
for an arbitrary state $\rho$.
Assume that its reverse operation, $U^{-1}=U^\dag$, is not an incoherent operation, then $C(\rho)=C(U^\dag U\rho U^\dag U)>C(U\rho U^\dag)$, which leads to a contradiction with Eq.~\eqref{eq:lemma:2}.
\end{proof}
\begin{lemma}\label{}
The coherence of an arbitrary state is invariant under incoherent unitary operations.
\end{lemma}
\begin{proof}
Consider an incoherent unitary operation $U$. Its reverse operation $U^\dag$ is also a unitary operation $U$ according to Lemma~\ref{lemma:2}. Then $C(\rho)=C(U^\dag U\rho U^\dag U)\le C(U\rho U^\dag)$. On the other hand, $C(U\rho U^\dag)\le C(\rho)$ since $U$ is an incoherent operation, which leads to $C(\rho)=C(U\rho U^\dag)$.
\end{proof}
With the lemmas above, we can first prove that the distance-based measure in Eq.~\eqref{eq:distancemeasure} satisfy all the conditions of a BD-discord measure.
\begin{proof}
Proof of (BD1). It is straightforward that $BD_{J_A}(\sigma^{QI}_{AB})=0$ according to the definition.

Proof of (BD2). We have such relations

\begin{equation}
\begin{aligned}
&BD_{J_A}[\Lambda_{SQI}(\rho_{AB})] \\
& =\min_{\sigma^{IQ}_{AB}\in \delta_{IQ}}d(\Lambda_{SQI}(\rho_{AB})|| \sigma^{IQ}_{AB})\\
&=\min_{\sigma^{IQ}_{AB}\in \delta_{IQ}}d(\Lambda_{SQI}(\rho_{AB})|| \Lambda_{SQI}(\sigma^{IQ}_{AB})) \\
&\leq \min_{\sigma^{IQ}_{AB}\in \delta_{IQ}} d(\rho_{AB}|| \sigma^{IQ}_{AB}) \\
&=BD_{J_A}(\rho_{AB}),
\end{aligned}
\end{equation}
where the second equality is because $\Lambda_{SQI}(\sigma^{IQ}_{AB})\in \delta_{IQ}$ and the inequality is due to the contractive nature of a distance measure, i.e, the distance will not increase under a completely positive and trace preserving (CPTP) map.

Proof of (BD3). Note the local incoherent unitary operation on $A$ and a unitary operation on $B$ as $U^I_A \otimes U_B$, which is a SQI operation, then
\begin{equation}
\begin{aligned}
	&\min_{\sigma^{IQ}_{AB}\in \delta_{IQ}}d(U^I_A \otimes U_B\rho_{AB}U^{I\dag}_A \otimes U^{\dag}_B||\sigma^{IQ}_{AB}) \\
	&\leq\min_{\sigma^{IQ}_{AB}\in \delta_{IQ}}d(\rho_{AB}||\sigma^{IQ}_{AB})
\end{aligned}
\end{equation}

On the other hand, the reverse operation $U^{I\dag}_A \otimes U^\dag_B$ is also a SQI operation since $U^{I\dag}_A$ is an incoherent operation according to Lemma \ref{lemma:2}, then
\begin{equation}
\begin{aligned}
	&\min_{\sigma^{IQ}_{AB}\in \delta_{IQ}}d(\rho_{AB}||\sigma^{IQ}_{AB})\\
	&=\min_{\sigma^{IQ}_{AB}\in \delta_{IQ}}d(U^I_AU^{I\dag}_A \otimes U_BU^{\dag}_B\rho_{AB}U^{I\dag}_AU^I_A \otimes U^{\dag}_BU_B||\sigma^{IQ}_{AB}) \\
	&\leq \min_{\sigma^{IQ}_{AB}\in \delta_{IQ}}d(U^I_A \otimes U_B\rho_{AB}U^{I\dag}_A \otimes U^{\dag}_B||\sigma^{IQ}_{AB})
\end{aligned}
\end{equation}

Thus we conclude that
\begin{equation}
\begin{aligned}
	&\min_{\sigma^{IQ}_{AB}\in \delta_{IQ}}d(\rho_{AB}||\sigma^{IQ}_{AB})\\
	&=\min_{\sigma^{IQ}_{AB}\in \delta_{IQ}}d(U^I_A \otimes U_B\rho_{AB}U^{I\dag}_A \otimes U^{\dag}_B||\sigma^{IQ}_{AB})
\end{aligned}
\end{equation}

\end{proof}

Next we prove that the convex roof measure Eq.~\eqref{eq:convexmeasure} also satisfies all conditions of a BD-discord measure.
\begin{proof}
Proof of (BD1). Consider the spectral decomposition of $\rho_B^{j_A}$ in Eq.~\eqref{eq:iqstate1}, an IQ state can be rewritten as
\begin{equation}
\sigma_{AB}^{IQ}=\sum_{j_A, l_{j_A}}p_{j_A}l_{j_A}\ket{j_Al_{j_A}}\bra{j_Al_{j_A}}.
\end{equation}	
\end{proof}
For each pure state component $\ket{j_Al_{j_A}}$, the local randomness is zero according to Eq.~\eqref{Eq:localrandomness}. And such a decomposition is an optimal decomposition due to the non-negativity of a BD-discord measure.

Proof of (BD2). Suppose the optimal decomposition is $\rho_{AB}=\sum_e p_e \ket{\psi_{AB}}\bra{\psi_{AB}}_e$. For an arbitrary component $\ket{\psi_{AB}}_e$, the local randomness is
\begin{equation}\label{Eq:localrandomnessofpsiab}
R(\ket{\psi_{AB}}_e)=S(\sum_{j_A}|\langle{j_A}|\mathrm{tr_B}\ket{\psi_{AB}}|^2 \ket{j_A}\bra{j_A})
\end{equation}
We notice that Eq.~\eqref{Eq:localrandomnessofpsiab} is equal to the relative entropy of BD-discord measure of $\ket{\psi_{AB}}$, which is a distance-based measure and contractive under $\Lambda_{SQI}$. The convex roof is a mixture of the local randomness for each pure state component, and the mixture is also contractive under $\Lambda_{SQI}$.

Proof of (BD3). Same as the proof for distance-based measure.

\section{Proof of Theorem~\ref{Theorem:discord}}\label{app:proofdiscord}
\begin{proof}
Proof of (D1).	For a classical-quantum state in Eq.~\eqref{Eq:CQstate}, we set $U_A\ket{n}=\ket{j_A}$ for $n=1,2,\cdots d_A$, then
\begin{equation}
\begin{aligned}
&BD_{J_A}\left[U_A\otimes I  \left(\sum^{d_A}_n p_n\ket{n}\bra{n}\otimes\rho_B^n \right) U^\dag_A \otimes I\right] \\
&= BD_{J_A}\left[\sum^{d_A}_n p_n (U_A \ket{n})(\bra{n}U_A^\dag) \otimes \rho_B^n \right] \\
&= BD_{J_A}\left(\sum^{d_A}_{j_A} p_{j_A} \ket{j_A}\bra{j_A} \otimes \rho_B^n \right) \\
&= 0.
\end{aligned}
\end{equation}
We can see that such a $U_A$ is optimal, which realizes a minimization of $D(\rho_{AB})$ due to the non-negativity of a basis-dependent discord measure.

Proof of (D2). Since $I\otimes \Phi_{B} \subset \Lambda_{SQI}$, from (BD2) we have
\begin{equation}
BD_{J_A}[I\otimes \Phi_{B}(\rho_{AB})]\leq BD_{J_A}(\rho_{AB})
\end{equation}
and their minimization on the local basis also satisfies
\begin{equation}
\begin{aligned}
	&\min_{U_A}BD_{J_A}[I\otimes \Phi_{B}(U_A\otimes I \rho_{AB}U^\dag_A\otimes I)] \\
	&\leq\min_{U_A} BD_{J_A}(U_A\otimes I \rho_{AB}U^\dag_A\otimes I)	
\end{aligned}
\end{equation}

Proof of (D3).  Our target is to prove
\begin{equation}\label{eq:invariant}
\begin{aligned}
&\min_{U_A}BD_{J_A}[(U_A\otimes U_B)(U_A\otimes I)\rho_{AB}(U^{\dag}_A\otimes I)(U_A^\dag \otimes U_{B}^\dag)] \\
&=\min_{U_A} BD_{J_A}[(U_A\otimes I)\rho_{AB}(U^\dag_A\otimes
I)]	
\end{aligned}
\end{equation}
Note that, in our definition of discord, the minimization is over all local basis, it is equal to prove that
\begin{equation}
\begin{aligned}
&\min_{U_A}BD_{J_A}[(I\otimes U_B)(U_A\otimes I)\rho_{AB}(U^{\dag}_A\otimes I)(I \otimes
U_{B}^\dag)] \\
&=\min_{U_A} BD_{J_A}[(U_A\otimes I)\rho_{AB}(U^\dag_A\otimes
I)]	
\end{aligned}
\end{equation}

Since $I\otimes U_B \subset I\otimes \Phi_{B}$, according to (D2) we have
\begin{equation}\label{eq:}
\begin{aligned}
&\min_{U_A}BD_{J_A}[(I\otimes U_B)(U_A\otimes I)\rho_{AB}(U^{\dag}_A\otimes I)(I \otimes
U_{B}^\dag)] \\
&\leq \min_{U_A} BD_{J_A}[(U_A\otimes I)\rho_{AB}(U^\dag_A\otimes
I)]	
\end{aligned}
\end{equation}
Apply local operation $I\otimes U_{B}^\dag$ on both sides,
\begin{equation}
\begin{aligned}
&\min_{U_A}BD_{J_A}[(U_A\otimes I)\rho_{AB}(U^{\dag}_A\otimes I)] \\
&\leq \min_{U_A} BD_{J_A}[(I\otimes U_{B}^\dag)(U_A\otimes I)\rho_{AB}(U^\dag_A\otimes
I)(I\otimes U_B)]	
\end{aligned}
\end{equation}
As the local operation $I\otimes U_{B}^\dag\subset I\otimes \Phi_B$, we also have
\begin{equation}
\begin{aligned}
&\min_{U_A}BD_{J_A}[(U_A\otimes I)\rho_{AB}(U^{\dag}_A\otimes I)] \\
&\leq \min_{U_A} BD_{J_A}[(I\otimes U_{B}^\dag)(U_A\otimes I)\rho_{AB}(U^\dag_A\otimes
I)(I\otimes U_B)]	
\end{aligned}
\end{equation}
Thus we prove Eq.~\eqref{eq:invariant}.
\end{proof}

\section{Proof of Theorem~\ref{Theorem:entanglement}}\label{app:proofentanglement}
\begin{proof}
Since condition (E2a) can be derived with (E2b) and (E3),
\begin{equation}
\begin{aligned}
E(\Lambda_{LOCC}(\rho_{AB}))&=E(\sum p_n \rho_{AB}^n)  \\
&\overset{C3}\leq \sum p_n E(\rho_{AB}^n) \\
&\overset{C2b} \leq E(\rho_{AB}),
\end{aligned}
\end{equation}
where $\rho_{AB}^n=\hat{K}_n\rho_{AB}\hat{K}_n^\dag/p_n$ and $p_n=Tr({\hat{K}_n\rho_{AB}\hat{K}_n^\dag} )$,
we only need to prove (E1), (E2b), (E3) and (E4).

Proof of (E1). Since the set of separable states is convex and closed, a separable state $\sigma_{AB}=\sum_j p_j \rho_A^j \otimes \rho_B^j$ can always be expressed as a mixture of pure separable states, i.e., product states.
\begin{equation}\label{eq:sepdecomp}
\sigma_{AB}=\sum_j p_j \rho_A^j \otimes \rho_B^j=\sum_e p_e \ket{\psi_{A}}_e\ket{\psi_B}_e \bra{\psi_{A}}_e\bra{\psi_B}_e
\end{equation}

Substitute Eq.~\eqref{eq:sepdecomp} into Eq.~\eqref{eq:entanglementdefinition}, for each pure state component $\ket{\psi_{A}}_e\ket{\psi_B}_e$, we set a certain $U^e_A$ such that $U^e_A\ket{\psi_A}=\ket{j_A}$, then
\begin{widetext}
\begin{equation}
\begin{aligned}
 & \min_{p_e,\ket{\psi_{AB}}_e}\sum_e p_e \min_{U^e_A} BD_{J_A}(U^e_A\otimes I \ket{\psi_{A}}_e\ket{\psi_B}_e \bra{\psi_{A}}_e\bra{\psi_B}_e U^{\dag e}_A \otimes I) \\
 & =\min_{p_e,\ket{\psi_{AB}}_e}\sum_e p_e BD_{J_A}(\ket{j_A}\bra{j_A}\otimes \ket{\psi_B}_e\bra{\psi_B}) \\
 & =0
 \end{aligned}
\end{equation}
\end{widetext}
We can see that such a st of $U_A^e$ and decomposition are optimal, which realizes a minimization of $E(\rho_{AB})$ due to the non-negativity of a basis-dependent discord measure.

%Then $E(\rho_{AB})\leq \sum_k\lambda_k \min_{J_A}Q_{J_A}\left(\sum_{j=1}^d p_j \ket{n_{jk}}\bra{n_{jk}}\otimes \rho^B_j\right)=0$ since the lhs. should take the decomposition letting $E(\rho^{AB})$ be the minimum. Considering the entanglement measure should always be non-negative, $E(\rho_{AB})=0$.

Proof of (E3). Suppose an arbitrary decomposition of $\rho_{AB}=\sum_l p_l \rho^l_{AB}$, and
\begin{equation}\label{eq:entanglementconvexity}
\begin{aligned}
&\sum_l p_l E(\rho^l_{AB})= \\
&\sum_l p_l \min_{p_{e}}\sum_{e} p_{e}  \min_{U^e_A} BD_{J_A}(U^e_A\otimes I \ket{\psi_{AB}}_{e}\bra{\psi_{AB}}_e U^{\dag e}_A \otimes I)
\end{aligned}
\end{equation}
where we simplify the subscript of minimizing decomposition $p_e, \ket{\psi_{AB}}_e$ to $p_e$.
Suppose for each component $\rho_{AB}^l$ the optimal decomposition is $\rho_{AB}^l=\sum_{e_l}p_e^l \ket{\psi_{AB}}_e^l\bra{\psi_{AB}}_e^l$, and we can further rewrite Eq.~\eqref{eq:entanglementconvexity} as
\begin{equation}\label{eq:optimaldecomp1}
\begin{aligned}
&\sum_l p_l E(\rho^l_{AB})\\
&=\sum_l p_l \sum_{e_l} p_{e}^l \min_{U^e_A} BD_{J_A}(U^e_A\otimes I \ket{\psi_{AB}}_{e}^l\bra{\psi_{AB}}_e^l U^{\dag e}_A \otimes I)	\\
&=\sum_l\sum_{e_l} p_l p_e^l \min_{U^e_A} BD_{J_A}(U^e_A\otimes I \ket{\psi_{AB}}_{e}^l\bra{\psi_{AB}}_e^l U^{\dag e}_A \otimes I)
\end{aligned}
\end{equation}
Similarly we assume the optimal decomposition for $\rho_{AB}$ is $\rho_{AB}=\sum_e p_e \ket{\psi_{AB}}_e\bra{\psi_{AB}}_e$
\begin{equation}\label{eq:optimaldecomp2}
\begin{aligned}
& E(\rho_{AB})=\\
&\sum_e p_e \min_{U^e_A} BD_{J_A}(U^e_A\otimes I \ket{\psi_{AB}}_{e}\bra{\psi_{AB}}_e U^{\dag e}_A \otimes I)
\end{aligned}
\end{equation}

Compare Eq.~\eqref{eq:optimaldecomp1} and Eq.~\eqref{eq:optimaldecomp2}, we can see that they are all probabilistic mixture of bipartite pure state discord. However, the ways of decomposition in Eq.~\eqref{eq:optimaldecomp2} is more than those in Eq.~\eqref{eq:optimaldecomp1} since the latter is constrained by the decomposition $\rho_{AB}=\sum_l p_l \ket{\psi_{AB}}_e \bra{\psi_{AB}}_e$
%$Q_{J_A}(\cdot)$ is a coherence measure in arbitrary basis $\{J_A\}$, hence $\min_{J_A}Q_{J_A}(\cdot)$ is a convex function. For an arbitrary decomposition $\rho_{AB}=\sum_n p_n\rho_{AB}^n$, we have $\min_{J_A}Q_{J_A}(\rho_{AB})\leq \sum_n p_n (\rho_{AB}^n)$. Therefore an arbitrary mixture of $\min_{J_A}Q_{J_A}(\cdot)$, such as $E(\cdot)$, will also be a convex function.
Then we conclude that
\begin{equation}
E(\rho_{AB})\geq \sum_l p_l E(\rho^l_{AB})	
\end{equation}

Proof of (E2b). Suppose the decomposition of $\rho_{AB}$ that achieves minimum of $E(\rho_{AB})$ is $\rho_{AB}=\sum_e p_e \rho_{AB}^e$, where $\rho_{AB}^e$ is a pure state. After the CPTP channel of LOCC,
\begin{equation}
\begin{aligned}
\rho^n_{AB}&=\frac{\hat{K}_n\rho_{AB}\hat{K}_n^\dag}{p_n} \\
&=\sum_e \frac{p_e}{p_n}\hat{K}_n \rho_{AB}^e \hat{K}_n^\dag \\
&=\sum_e \frac{p_e}{p_n} p_{en} \rho^{en}_{AB}
\end{aligned}
\end{equation}
where $p_{en}=Tr(\hat{K}_n \rho_{AB}^e \hat{K}_n^\dag)$ and $\rho^{en}_{AB}=\hat{K}_n \rho_{AB}^e \hat{K}_n^\dag/p_{en}$. Then we have
\begin{equation}
\begin{aligned}
&E(\rho_{AB})
=\sum_e p_e \min_{U_A} BD_{J_A}(U_A\otimes I\rho_{AB}^e U_A^\dag \otimes I) \\
&\geq \sum_e p_e \min_{U_A} \sum_n p_{en} BD_{J_A}(U_A\otimes I\rho_{AB}^{en} U_A^\dag \otimes I ) \\
&\geq \sum_e p_e \sum_n p_{en} \min_{U_{An}} BD_{J_A}(U_{An}\otimes I\rho_{AB}^{en} U_{An}^\dag \otimes I) \\
&=\sum_n \sum_e p_e p_{en} \min_{U_{An}} BD_{J_A}(U_{An}\otimes I\rho_{AB}^{en} U_{An}^\dag \otimes I ) \\
&=\sum_n p_n\sum_e \frac{p_e p_{en}}{p_n} \min_{U_{An}} BD_{J_A} (U_{An}\otimes I\rho_{AB}^{en} U_{An}^\dag \otimes I)\\
&\geq \sum_n p_n E(\rho^n_{AB}),
\end{aligned}
\end{equation}
where the first inequality is due to the selective monotonicity of distance-based BD-discord measure,
\begin{equation}
BD_{J_A}(\rho_{AB})	\geq \sum_n p_n  BD_{J_A}(\rho^n_{AB})
\end{equation}
 the second inequality is because the minimization over local basis according to each component after channel $U_{An}$ is more powerful than an entire minimization $U_A$.

Proof of (E4). Local unitary operations $U_A \otimes U_B$ belong to LOCC. Then according to (E2a),
\begin{equation}
E(\rho_{AB})\geq E(U_A \otimes U_B \rho_{AB} U_A^\dag \otimes U_B^\dag)
\end{equation}
 Apply $U_A^\dag \otimes U_B^\dag$ to the last equation,
\begin{equation}
E(U_A^\dag \otimes U_B^\dag  \rho_{AB} U_A \otimes U_B) \geq E(\rho_{AB})
\end{equation}

On the other hand, operations $U_A^\dag \otimes U_B^\dag$ also belong to LOCC.
\begin{equation}
E(U_A^\dag \otimes U_B^\dag  \rho_{AB} U_A \otimes U_B) \leq E(\rho_{AB})
\end{equation}

Therefore we have
\begin{equation}
E(\rho_{AB})= E(U_A \otimes U_B \rho_{AB} U_A^\dag \otimes U_B^\dag)
\end{equation}
\end{proof}

\bibliography{bibmeasure}
\bibliographystyle{apsrev4-1}
%%%%%%%%%%%%%%%%%%%%%%%%%%%%%%%%%%%%%%%%
\end{document}